\documentclass[showpacs,pra, reprint, a4paper]{revtex4-1}
\usepackage{amsmath,amssymb,amsthm}
\usepackage{bm}

\newtheorem{theorem}{Theorem}
\newtheorem{lemma}[theorem]{Lemma}
\newtheorem{corollary}[theorem]{Corollary}
\theoremstyle{definition}
\newtheorem{definition}[theorem]{Definition}
\newtheorem{example}[theorem]{Example}

\theoremstyle{remark}
\newtheorem*{remark}{Remark}

\begin{document}
\title{Three-input majority function as the unique optimal function for the bias amplification using nonlocal boxes}
\date{\today}
\author{Ryuhei Mori}
\email{mori@c.titech.ac.jp}
\affiliation{School of Computing, Tokyo Institute of Technology, Tokyo 152-8552, Japan}

\begin{abstract}
Brassard et al.\@ [Phys.\@ Rev.\@ Lett.\@ \textbf{96}, 250401 (2006)]  showed that shared nonlocal boxes with the CHSH probability greater than $\frac{3+\sqrt{6}}6$ yields trivial communication complexity.
There still exists the gap with the maximum CHSH probability $\frac{2+\sqrt{2}}4$ achievable by quantum mechanics.
It is an interesting open question to determine the exact threshold for the trivial communication complexity.
Brassard et al.'s idea is based on the recursive bias amplification by the 3-input majority function.
It was not obvious if other choice of function exhibits stronger bias amplification.
We show that the 3-input majority function is the unique optimal, so that one cannot improve the threshold $\frac{3+\sqrt{6}}6$
 by Brassard et al.'s bias amplification.
In this work, protocols for computing the function used for the bias amplification are restricted to be non-adaptive protocols or particular adaptive protocol inspired by Paw\l owski et al.'s protocol for information causality [Nature \textbf{461}, 1101 (2009)].
We first show a new adaptive protocol inspired by Paw\l owski et al.'s protocol,
and then show that the new adaptive protocol is better than any non-adaptive protocol.
Finally, we show that the 3-input majority function is the unique optimal for the bias amplification if we apply the new adaptive protocol to each step of the bias amplification.
\end{abstract}


\maketitle

\section{Introduction}
Bell showed that quantum mechanics allows correlations that cannot be generated by classical physics~\cite{bell1964einstein}.
Clauser, Horne, Shimony, and Holt (CHSH) found simpler constraints on correlations which could be violated on quantum mechanics, but is always satisfied
on classical physics~\cite{PhysRevLett.23.880}, which in fact, characterize the set of correlations generated by classical physics on the binary setting~\cite{PhysRevLett.48.291}.
Apart from the concrete mathematical description of quantum mechanics, we can only consider abstract statistical behavior realized by quantum mechanics.
A nonlocal box is an abstract device which represents statistical behavior of separable measurements on a possibly entangled state on quantum mechanics and superquantum theory as well.
A nonlocal box is assumed to be shared by two parties, Alice and Bob.
A nonlocal box has input ports and output ports on the both sides.
A nonlocal box is specified by the conditional probability distribution $p(a,b\mid x,y)$ representing the probability of 
outputting $a$ to Alice and $b$ to Bob when Alice and Bob input $x$ and $y$ into the nonlocal box, respectively.
Here, all of $x$, $y$, $a$ and $b$ are assumed to be either of 0 or 1.
They cannot communicate by using the nonlocal box since it satisfies the no-signaling condition
\begin{align*}
\sum_{b\in\{0,1\}} p(a,b\mid x,0) &= \sum_{b\in\{0,1\}} p(a,b\mid x,1)\\
\sum_{a\in\{0,1\}} p(a,b\mid 0,y) &= \sum_{a\in\{0,1\}} p(a,b\mid 1,y).
\end{align*}
The CHSH probability $P_{\mathrm{CHSH}}$ is a measure of the nonlocality of the nonlocal box defined by
\begin{equation*}
P_{\mathrm{CHSH}} := \frac14\sum_{x,y} \sum_{\substack{a,b\\ a\oplus b = x\wedge y} }p(a,b\mid x,y).
\end{equation*}
While the maximum CHSH probability given by classical physics is $P_{\mathrm{CHSH}}=3/4$,
that for quantum mechanics is $P_{\mathrm{CHSH}}=\frac{2+\sqrt{2}}4$~\cite{cirelson1980}.
On the other hand, Popescu and Rohrich showed that there exists the nonlocal box, called the PR box, with $P_{\mathrm{CHSH}}=1$~\cite{PR1994}.
Hence, it is a natural question why quantum mechanics cannot achieve the CHSH probability greater than $\frac{2+\sqrt{2}}4$.
Van Dam showed that if Alice and Bob share unlimited number of PR boxes, they can compute arbitrary function $f(x,y)$ only by sending 1 bit to each other where $x$ and $y$ are $n$ bits owned by Alice and Bob, respectively~\cite{vandam2005}.
It gives the explanation why Nature does not allow $P_{\mathrm{CHSH}}=1$ since we strongly believe that the trivial communication complexity must not be allowed by Nature.
Furthermore, Brassard et al.\@ showed that the nonlocal box with $P_{\mathrm{CHSH}} > \frac{3+\sqrt{6}}6$ yields the trivial communication complexity on the probabilistic setting~\cite{PhysRevLett.96.250401}.
It has not been known whether or not the communication complexity is trivial when the CHSH probability is between $\frac{2+\sqrt{2}}4$ and $\frac{3+\sqrt{6}}6$.
Later, Paw\l owski et al.\@ completely characterized the quantum CHSH probability $\frac{2+\sqrt{2}}4$ by using new principle called information causality~\cite{pawlowski2009information}.
However, it is still interesting to determine the exact threshold of $P_{\mathrm{CHSH}}$ for the trivial communication complexity.

In this paper, we show that the trivial communication complexity below $\frac{3+\sqrt{6}}6$ cannot be proved by Brassard et al.'s technique.
Their technique is based on the recursive bias amplification from exponentially small bias to constant bias by using the 3-input majority function $\mathrm{Maj}_3$.
It was not obvious that $\mathrm{Maj}_3$ is the best choice for the bias amplification.
It seems to be curious that $\mathrm{Maj}_3$ is the optimal function for the bias amplification if it is true.
In this paper, we show that $\mathrm{Maj}_3$ is the unique optimal function for the bias amplification.

\begin{theorem}\label{thm:main}
The 3-input majority function is the unique optimal for Brassard et al.'s technique of the bias amplification using the nonlocal boxes.
Hence, one cannot obtain the threshold for the trivial communication complexity smaller than $\frac{3+\sqrt{6}}{6}$ by Brassard et al.'s technique.
\end{theorem}
In Brassard et al.'s protocol, the 3-input majority function $\mathrm{Maj}_3$ is computed by a non-adaptive protocol, i.e.,
inputs for nonlocal boxes are independent of outputs of other nonlocal boxes.
In this work, we introduce a new adaptive protocol inspired by~\cite{pawlowski2009information}, and show that the new adaptive protocol is no worse than arbitrary non-adaptive protocol.
Then, we show Theorem~\ref{thm:main} for generalizations of Brasssard et al.'s protocol in which arbitrary boolean function is used for the bias amplification in place of $\mathrm{Maj}_3$, and is
computed by the new adaptive protocol.
In this work, protocols for the computation of the function corresponding to $\mathrm{Maj}_3$
are restricted to be non-adaptive protocols or the new adaptive protocol inspired by~\cite{pawlowski2009information}.
For the proof of Theorem~\ref{thm:main}, we use the Fourier analysis of boolean functions developed in theoretical computer science~\cite{odonnell2014analysis}.

\section{Preliminaries}\label{sec:pre}
\subsection{XOR protocol and nonlocal boxes}
We introduce some notions and notations.

\begin{definition}
For a boolean function $f\colon\{0,1\}^n\times\{0,1\}^n\to\{0,1\}$,
XOR protocol with bias $\epsilon$ is a process of computations by Alice and Bob
in which Alice and Bob compute $a$ and $b$, respectively, by using nonlocal boxes and shared random bits but without any communication
such that $a\oplus b=f(x,y)$ with probability $(1+\epsilon)/2$.
\end{definition}

There is a simple XOR protocol with bias $2^{-n}$ for arbitrary function~\cite{PhysRevLett.96.250401}.
\begin{lemma}
\label{lem:1bit2n}
There is an XOR protocol with bias $2^{-n}$ for arbitrary function $f(x,y)$ without using nonlocal boxes.
\end{lemma}
\begin{proof}
Let $r\in\{0,1\}^n$ be shared uniform random bits.
Let $a=f(x,r)$.
Let $b = 0$ if $r=y$ and $b=r'$ otherwise where $r'\in\{0,1\}$ is Bob's private uniform random bit.
Then, $a\oplus b =f(x,y)$ with probability $\frac12+\frac1{2^{n+1}}$.
\end{proof}
\begin{definition}
The nonlocal box is said to be isotropic if
\begin{equation*}
\sum_{\substack{a,b\\ a\oplus b = x\wedge y}} p(a,b\mid x,y)
\end{equation*}
 does not depend on $x$ and $y$ and if the marginal distributions for $a$ and $b$ are uniform for any $x$ and $y$.
\end{definition}
It was shown in~\cite{PhysRevA.73.012112, PhysRevLett.102.180502} that the isotropic nonlocal box can be simulated by arbitrary nonlocal box with the same CHSH probability.
\begin{lemma}\label{lem:iso}
Using arbitrary given nonlocal box, the isotropic nonlocal box with the same CHSH probability can be simulated.
\end{lemma}
From Lemma~\ref{lem:iso}, in this study, we assume that all nonlocal boxes are isotropic.
Forster et al.\@ showed that non-isotropic nonlocal boxes can be used for the nonlocality distillation, which is the
amplification of the CHSH probability~\cite{PhysRevLett.102.120401}.
Brunner and Skrzypczyk showed that there exists non-isotropic nonlocal box with $P_{\mathrm{CHSH}}=3/4+\epsilon$
for arbitrary small $\epsilon>0$ which allows the simulation of nonlocal box arbitrarily close to the PR box~\cite{PhysRevLett.102.160403}.
Of course, such nonlocal box cannot be simulated on quantum mechanics even if the CHSH probability of the nonlocal box is achievable by quantum mechanics.
In this study, we do not consider the nonlocality distillation, but consider the XOR protocol using isotropic nonlocal boxes.

\subsection{Fourier analysis}
The Fourier analysis is the main mathematical tool in this work.
\begin{definition}
Any boolean function $f\colon\{+1,-1\}^n\to\{+1,-1\}$ can be represented by a polynomial on $\mathbb{R}$ uniquely
\begin{equation*}
f(x) = \sum_{S\subseteq [n]} \widehat{f}(S)\prod_{i\in S} x_i
\end{equation*}
where $[n]:=\{1,2,\dotsc,n\}$.
Here, $(\widehat{f}(S))_{S\subseteq[n]}$ are called the Fourier coefficients of $f$.
When we consider the Fourier coefficients of boolean function $f\colon\{0,1\}^n\to\{0,1\}$,
we regard $f$ as the function from $\{+1,-1\}^n$ to $\{+1,-1\}$.
From Parseval's identity, the sum of squares of the Fourier coefficients is 1.

Let $\mathrm{supp}(\widehat{f}):=\{S\subseteq[n]\mid \widehat{f}(S)\ne 0\}$.
For $S\subseteq[n]$, let $1_S$ be a vector on $\mathbb{F}_2$ of length $n$ such that $i$-th element of $1_S$ is 1 iff $i\in S$.
Let $\dim(\widehat{f})$ be the Fourier dimension of $f$ which is the dimension of linear space on $\mathbb{F}_2$ spanned by $\{1_S\mid S\in \mathrm{supp}(\widehat{f})\}$.
\end{definition}

\subsection{One-way communication complexities}
We introduce notions on the one-way communication complexity of $f\colon\{0,1\}^n\times\{0,1\}^n\to\{0,1\}$.
Let $M_f$ be a $2^n\times 2^n$ matrix whose $(x,y)$-element is $f(x,y)$.
Let $D_{\rightarrow}(f)$ be the one-way communication complexity of $f$ from Alice to Bob, which is the minimum $m$ such that
there exist functions $s\colon \{0,1\}^n\to\{0,1\}^m$ and $h\colon\{0,1\}^m\times\{0,1\}^n\to\{0,1\}$ satisfying the identity $f(x,y) = h(s(x), y)$.
Similarly, let $D_{\leftarrow}(f)$ be the one-way communication complexity of $f$ from Bob to Alice.
The one-way communication complexities can be characterized by the number of distinct rows and columns of $M_f$, i.e,
$D_\rightarrow(f) = \lceil \log_2 \mathrm{nrows}(M_f)\rceil$ and
$D_\leftarrow(f) = \lceil \log_2 \mathrm{ncols}(M_f)\rceil$
where $\mathrm{nrows}(M_f)$ and $\mathrm{ncols}(M_f)$ denote the number of distinct rows and the number of distinct columns of $M_f$, respectively.
We also define
\begin{align*}
D_{\rightarrow}^\oplus(f) &:= \min_{A\colon \{0,1\}^n\to\{0,1\}} D_{\rightarrow}\left(f(x,y)\oplus A(x)\right)\\
D_{\leftarrow}^\oplus(f) &:= \min_{B\colon \{0,1\}^n\to\{0,1\}} D_{\leftarrow}\left(f(x,y)\oplus B(y)\right).
\end{align*}
Here, $D_\rightarrow^\oplus(f)$ is the minimum number of bits Alice have to send to Bob such that
Alice can compute $a$ and Bob can compute $b$ satisfying $a\oplus b = f(x,y)$.

\subsection{Other notations}
For odd $n$, let $\mathrm{Maj}_{n}\colon\{0,1\}^n\to\{0,1\}$ be the majority function on $n$ variables.
For even $n$, let $\mathrm{Maj}_{n}$ be the set of majority functions on $n$ variables where the definitions for the tie cases are arbitrary.
Since there are $\binom{n}{\frac{n}2}$ tie cases, $|\mathrm{Maj}_n|=2^{\binom{n}{\frac{n}2}}$ for even $n$.
Note that a function $f\colon\{0,1\}^{2k}\to\{0,1\}$ which ignores one of the $2k$ input variables, and outputs the majority of the other $2k-1$ variables is a member of $\mathrm{Maj}_{2k}$.
Finally, let $\delta:=2P_{\mathrm{CHSH}}-1$, i.e., $P_{\mathrm{CHSH}}=\frac{1+\delta}2$.
Here, we call $\delta$ the bias of the CHSH probability.

\section{Brassard et al.'s protocol}\label{sec:brassard}
Brassard et al.'s basic idea is bias amplification by $\mathrm{Maj}_3$.
They showed that $\mathrm{Maj_3}$ can be computed by using two PR boxes.
Here, we give a simple argument showing that two PR boxes are sufficient to compute $\mathrm{Maj}_3(x\oplus y)$.
The $\mathbb{F}_2$-polynomial representation of the 3-input majority function is $\mathrm{Maj}_3(z_1,z_2,z_3)=z_1z_2\oplus z_2z_3\oplus z_3z_1$.
Hence, one obtains the representation
\begin{align}
&\mathrm{Maj}_3(x_1\oplus y_1, x_2\oplus y_2, x_3\oplus y_3)\nonumber\\
&=
(x_1\oplus x_2)(y_2\oplus y_3) \oplus (x_2\oplus x_3)(y_1\oplus y_2)\nonumber\\
&\quad\oplus (x_1x_2 \oplus x_2x_3 \oplus x_3 x_1)
\oplus (y_1y_2 \oplus y_2y_3 \oplus y_3 y_1).
\label{eq:maj3}
\end{align}
The following is the protocol for computing $a$ and $b$.
Alice and Bob can compute their local terms 
$\mathrm{Maj}_3(x):=x_1x_2 \oplus x_2x_3 \oplus x_3 x_1$
and
$\mathrm{Maj}_3(y):=y_1y_2 \oplus y_2y_3 \oplus y_3 y_1$
without communication, respectively.
For the each of first two terms in~\eqref{eq:maj3}, they use the PR box.
For the first PR box, Alice and Bob input $x_1\oplus x_2$ and $y_2\oplus y_3$
and obtain $a_1$ and $b_1$, respectively.
Similarly,
for the second PR box, Alice and Bob input $x_2\oplus x_3$ and $y_1\oplus y_2$
and obtain $a_2$ and $b_2$, respectively.
Then, Alice and Bob output $a:= \mathrm{Maj}_3(x)\oplus a_1\oplus a_2$ and $b:= \mathrm{Maj}_3(y)\oplus b_1\oplus b_2$, respectively.
This is the XOR protocol without error using two PR boxes.
Von Neumann showed that the probability of correctness of computations sufficiently close to 1/2 is amplified by noisy $\mathrm{Maj}_3$ iff the computation of $\mathrm{Maj}_3$ succeeds with probability greater than 5/6~\cite{neumann1956}.
Hence, the threshold of the above protocol is given by the condition $P_{\mathrm{CHSH}}^2 + (1-P_{\mathrm{CHSH}})^2 > 5/6\iff P_{\mathrm{CHSH}}>\frac{3+\sqrt{6}}6$.
On this condition, the iterative applications of $\mathrm{Maj}_3$ to independent samples obtained by the protocol in Lemma~\ref{lem:1bit2n}
 give a constant bias.

Brassard et al.\@ invented the above elegant protocol, and showed that if $P_{\mathrm{CHSH}}>\frac{3+\sqrt{6}}6$, there exists an XOR protocol with constant bias for arbitrary function $f$.
However, there is no reason why $\mathrm{Maj}_3$ should be used for the bias amplification.
We can use arbitrary functions, e.g., the majority function on 5 variables, in place of $\mathrm{Maj}_3$.
Of course, on given number $n$ of input variables, the majority functions $\mathrm{Maj}_n$ minimize the threshold value, corresponding to 5/6 for $\mathrm{Maj}_3$.
However, non-majority function may require smaller number of nonlocal boxes than the majority functions.
Hence, non-majority functions are also candidates for the generalization of Brassard et al's protocol.
We have to generalize two quantities ``2'' and ``5/6'' in the case of $\mathrm{Maj}_3$, which are the number of nonlocal boxes needed for the computation
and the threshold for the probability of the correctness of computation of the function for the bias amplification, respectively.
In this work, these two quantities are clearly characterized.

Although we can consider general function $f\colon\{0,1\}^n\times\{0,1\}^n\to\{0,1\}$ in place of $\mathrm{Maj}_3(x\oplus y)$,
in this study, we restrict $f$ to be XOR function, i.e., $f(x,y)=g^\oplus(x,y):=g(x\oplus y)$ for some $g\colon\{0,1\}^n\to\{0,1\}$.
It seems to be a natural restriction since the inputs $x$ and $y$ have meaning only when their XOR is taken.
Linden et al.\@ showed that quantum mechanics has no advantage on XOR protocol for computation of XOR function when the input distribution is also XOR function~\cite{PhysRevLett.99.180502}.

\section{Non-adaptive PR-correct protocols}
Brassard et al.\@ consider the protocol according to the $\mathbb{F}_2$-polynomial representation~\eqref{eq:maj3} for computing $\mathrm{Maj}_3^\oplus$.
In this section, we show that for arbitrary given $f\colon \{0,1\}^n\times\{0,1\}^n\to\{0,1\}$, this protocol is the best protocol for computing $f(x,y)$
among all protocols satisfying the non-adaptivity and the PR-correctness.

\begin{definition}
An XOR protocol is said to be non-adaptive if inputs for nonlocal boxes does not depend on outputs of other nonlocal boxes.
An XOR protocol is said to be PR-correct if the protocol computes the target function $f(x,y)$ without error when the nonlocal boxes are PR boxes.
An XOR protocol is said to be non-redundant if the inputs $(l_i(x),r_i(x))_{i=1,\dotsc,t}$  for the nonlocal box satisfy
\begin{align}
&A(x) \oplus B(y) \oplus \bigoplus_{i=1}^{t} \left(C_i\wedge l_i(x)\wedge r_i(y)\right) = 0\nonumber\\
& \iff (C_i)_{i=1,\dotsc,t} = 0, A(x)=B(y).
\label{eq:red}
\end{align}
\end{definition}

The following lemma was shown by Kaplan et al.~\cite{kaplan2011nonlocal}.
Here, we give a short proof using Fourier analysis.
\begin{lemma}\label{lem:parity}
The outputs of both players in non-adaptive PR-correct non-redundant protocol must be parity of the outputs of nonlocal boxes and a function of local inputs.
\end{lemma}
\begin{proof}
Let $l_1(x),\dotsc,l_t(x)$ and $r_1(y),\dotsc,r_t(y)$ be the inputs of nonlocal boxes from Alice and Bob, respectively.
Let $a_1,\dotsc,a_t$ and $b_1,\dotsc,b_t$ be the outputs of the nonlocal boxes for Alice and Bob, respectively.
From any protocol, one can obtain a modified protocol using $(l'_i(x):=l_i(x)\oplus l_i(0),r'_i(x):=r_i(y)\oplus r_i(0))_{i=1,\dotsb,t}$ as the inputs for nonlocal boxes
since replacements of $a_i$ and $b_i$ by $a'_i\oplus l'_i(x)r_i(0) \oplus l_i(0)r_i(0)$ and $b'_i\oplus l_i(0)r'_i(y)$ for $i=1,\dotsc,t$, respectively, simulate the original protocol where $(a'_i, b'_i)_{i=1,\dotsc t}$ is the outputs of the nonlocal boxes in the modified protocol.
This transformation preserves non-adaptivity, PR-correctness and non-redundancy.
This transformation also preserves whether or not the outputs of both players are parity of the outputs of nonlocal boxes and a function of local inputs.
Hence, without loss of generality, we can assume that $l_1(0)=\dotsb=l_t(0)=r_1(0)=\dotsb=r_t(0)=0$.
Assume that $a=u_x(a_1,\dotsc,a_t)$ and $b=v_y(b_1,\dotsc,b_t)$.
Since the protocol is PR-correct, $a\oplus b = u_x(a_1,\dotsc,a_t)\oplus v_y(a_1\oplus z_1(x,y),\dotsc,a_t\oplus z_t(x,y))$ must be constant for all $(a_1,\dotsc,a_t)\in\{0,1\}^t$
where $z_i(x,y):=l_i(x)\wedge r_i(y)$.
By letting $x=0$ ($y=0$), we obtain that $v_y$ $(u_x)$ is equal to $u_0$ ($v_0$) or its negation for any $y$ ($x$), respectively.
Hence, there exists boolean functions $F\colon\{0,1\}^t\to\{0,1\}$, $\varphi, \psi\colon \{0,1\}^n\to\{0,1\}$ such that 
$u_x(a_1,\dotsc,a_t) = \varphi(x)\oplus F(a_1,\dotsc,a_t)$ and
$v_y(b_1,\dotsc,b_t) = \psi(y)\oplus F(b_1,\dotsc,b_t)$.
On the other hand, it holds on the $\{+1,-1\}$ domain that
\begin{align}
&ab = \left(\sum_{S\subseteq [t]} \widehat{u}_x(S) \prod_{i\in S} a_i\right)
\left(\sum_{S\subseteq [t]} \widehat{v}_y(S) \prod_{i\in S} (a_i z_i(x,y))\right)\nonumber\\
&= \sum_{S_1, S_2\subseteq [t]} \widehat{u}_x(S_1) \widehat{v}_y(S_2) \prod_{i\in S_2} z_i(x,y) \prod_{i\in (S_1\cup S_2)-(S_1\cap S_2)} a_i\nonumber\\
&= \sum_{S\subseteq [t]} \left(\sum_{\substack{S_1, S_2,\\ (S_1\cup S_2) - (S_1\cap S_2)= S}} \widehat{u}_x(S_1) \widehat{v}_y(S_2) \prod_{i\in S_2} z_i(x,y)\right)\nonumber\\
&\qquad \cdot\prod_{i\in S} a_i
\label{eq:PRc}
\end{align}
This is the Fourier expansion of
$u_x(a_1,\dotsc,a_t)\oplus v_y(a_1\oplus z_1(x,y),\dotsc,a_t\oplus z_t(x,y))$ as a function of $a_1,\dotsc,a_t$.
Since the function must be constant, the Fourier coefficients for the empty set must be $\pm1$, i.e.,
\begin{align*}
\sum_{S_1\subseteq[t]} \varphi(x)\psi(y)\widehat{F}(S_1)^2 \prod_{i\in S_1} z_i(x,y) \in \{+1,-1\}
\end{align*}
for any $x,y\in\{0,1\}^n$.
Hence, for any $x,y\in\{0,1\}^n$, $\prod_{i\in S_1}z_i(x,y)$ must be common
for all $S_1 \in\mathrm{supp}(\widehat{F})$.
The equality $\prod_{i\in S_1} z_i(x,y)=\prod_{i\in S_2} z_i(x,y)$ for $S_1\ne S_2$ implies $\prod_{i\in (S_1\cup S_2)-(S_1\cap S_2)} z_i(x,y)=1$
that means the existence of a redundant nonlocal box.
Hence, $\widehat{F}(S_1)\ne 0$ for unique $S_1\subseteq [t]$.
It implies that $u_x$ ($v_y$) is the parity of variables in $S_1$ and $\varphi(x)$ ($\psi(y)$), respectively.
\end{proof}
Naturally, we can ask whether or not the non-redundancy is restriction, i.e., whether or not we can reduce the error probability of the protocol by using the redundancy when the nonlocal boxes are not the PR boxes.
The following lemma says that redundancy does not help to reduce the error probability of non-adaptive PR-correct protocol.

\begin{lemma}\label{lem:redundant}
For arbitrary given non-adaptive PR-correct protocol, there exists non-adaptive PR-correct non-redundant protocol whose error probability is at most that of the original protocol for any bias $\delta$ of the CHSH probability.
\end{lemma}
\begin{proof}
As in the proof of Lemma~\ref{lem:parity}, we can assume without loss of generality that $l_1(0)=\dotsb=l_t(0)=r_1(0)=\dotsb=r_t(0)=0$.
Similarly to~\eqref{eq:PRc}, when the nonlocal boxes are not necessarily the PR boxes, $ab$ is equal to
\begin{align*}
&\sum_{S\subseteq [t]} \left(\sum_{\substack{S_1, S_2,\\ (S_1\cup S_2) - (S_1\cap S_2)= S}} \widehat{u}_x(S_1) \widehat{v}_y(S_2) \prod_{i\in S_2} (e_iz_i(x,y))\right)\\
&\quad\cdot \prod_{i\in S} a_i
\end{align*}
where $e_i$ represents the error of the output of $i$-th nonlocal box, i.e., $e_i=+1$ if the $i$-th nonlocal box computes correctly and $e_i=-1$ otherwise.
Recall that the bias of the CHSH probability is $\delta$, i.e., the expectation of $e_i$ is $\delta$.
Since the nonlocal boxes are isotropic, $e_i$ is independent of any other variables $x$, $y$, $(a_j)_{j\in[t]}$ and $(e_j)_{j\in[t]\setminus\{i\}}$ for $i\in[t]$.
Since the nonlocal boxes are isotropic, $a_i$ is uniformly distributed for all $i\in[t]$.
Hence, the expectation of $ab$ (the bias of $a\oplus b$) is
\begin{align*}
&\sigma(x,y) \varphi(x)\psi(y)\sum_{S_1\subseteq[t]}\widehat{F}(S_1)^2 \delta^{|S_1|}\\
&=:\sigma(x,y)\varphi(x)\psi(y) \mathbf{Stab}_\delta(F)
\end{align*}
where $\sigma(x,y)$ denotes the common sign of $\prod_{i\in S_1}z_i(x,y)\in\{+1,-1\}$
 for all $S_1\in \mathrm{supp}(\widehat{F})$.
Since the protocol is PR-correct, $\sigma(x,y)\sigma(x)\psi(y)\in\{+1,-1\}$ must be equal to $f(x,y)$.
Hence, the output of the protocol is correct with probability $(1+\mathbf{Stab}_{\delta}(F))/2$.
On the other hand, since $\prod_{i\in S}z_i(x,y)\in\{+1,-1\}$ is common for all $S\in\mathrm{supp}(\widehat{F})$,
we can obtain a new non-adaptive PR-correct protocol by replacing
$u_x(a_1,\dotsc,a_t)$ and $v_y(b_1,\dotsc,b_t)$ by $\varphi(x)\oplus \bigoplus_{i\in S^*} a_i$ and $\psi(y)\oplus \bigoplus_{i\in S^*} b_i$
for $S^*:=\mathop{\mathrm{argmin}}_{S\in\mathrm{supp}(\widehat{F})} |S|$, respectively.
In order to obtain non-adaptive PR-correct non-redundant protocol,
we shrink the set $S^*$ to $T\subseteq S^*$ if $S^*$ includes the redundancy
(The local terms $\varphi(x)$ and $\psi(y)$ should also be modified according to the shrinkage).
The bias of the probability of correctness of the protocol is $\delta^{|T|}\ge \delta^{|S^*|}\ge \mathbf{Stab}_{\delta}(F)$.
\end{proof}

Lemma~\ref{lem:redundant} implies that if we are interested in the minimization of the error probability among all non-adaptive PR-correct protocols,
we only have to consider non-adaptive PR-correct non-redundant protocols.

\section{The number of nonlocal boxes}
Lemma~\ref{lem:parity} implies that arbitrary non-adaptive PR-correct non-redundant protocol corresponds to $\mathbb{F}_2$-polynomial representation of $f(x,y)$
\begin{equation}
f(x,y) = A(x) \oplus B(y) \oplus \bigoplus_{i=1}^t l_i(x) r_i(y).
\label{eq:F2expand}
\end{equation}
Since the bias of the correctness of the corresponding protocol is $\delta^t$,
we define the following measure of the complexity.
\begin{definition}
The nonlocal box complexity $\mathrm{NLBC}(f)$ is the minimum $t$ such that there exists a representation~\eqref{eq:F2expand}.
\end{definition}
The nonlocal box complexity can be characterized by the rank of some matrix on $\mathbb{F}_2$.
The following theorem slightly generalizes a theorem in~\cite{kaplan2011nonlocal}.

\begin{theorem}\label{thm:f2rank}
For any $f\colon\{0,1\}^n\times\{0,1\}^n\to\{0,1\}$,
\begin{equation*}
\mathrm{NLBC}(f) = \mathrm{rank}_{\mathbb{F}_2}(M_{f'})
\end{equation*}
where $f'(x,y) = f(x,y) \oplus f(x,0) \oplus f(0,y) \oplus f(0,0)$,
and where $M_{f'}$ is a $2^n\times 2^n$ matrix on $\mathbb{F}_2$ such that its $(x,y)$-element is equal to $f'(x,y)$.
\end{theorem}
\begin{proof}
First, we show
$\mathrm{NLBC}(f) \le \mathrm{rank}_{\mathbb{F}_2}(M_{f'})$.
If $\mathrm{rank}_{\mathbb{F}_2}(M_{f'})=r$, there is a matrix factorization $M_{f'}=UV$
for some $2^n\times r$ matrix $U$ and $r\times 2^n$ matrix $V$.
It implies that $f'(x,y)=\bigoplus_{i=1}^r a_i(x) b_i(y)$ where $a_i(x)$ denotes $(x,i)$ element of $U$ and where $b_i(y)$ denotes $(i,y)$ element of $V$.
Hence, it holds $f(x,y) = \left(f(x,0)\oplus f(0,0)\right) \oplus f(0,y) \oplus \bigoplus_{i=1}^r a_i(x)b_i(y)$, and hence $\mathrm{NLBC}(f)\le r$.

Conversely,
if $\mathrm{NLBC}(f)=t$, there is a representation 
$f(x,y) = A(x) \oplus B(y) \oplus \bigoplus_{i=1}^t l_i(x) r_i(y)$.
There also exists a representation
$f'(x,y) = A'(x) \oplus B'(y) \oplus \bigoplus_{i=1}^t l_i(x) r_i(y)$.
Since $f'(0,y)=f'(x,0)=0$ for all $x$ and $y$,
by expanding constant terms in $l_i(x)$ and $r_i(x)$, we obtain a representation
$f'(x,y) = \bigoplus_{i=1}^t l'_i(x) r'_i(y)$.
It implies that there is a matrix factorization $M_{f'}=UV$
for $2^n\times r$ matrix $U$ and $r\times 2^n$ matrix $V$
where $(x,i)$ element of U is $l'_i(x)$ and $(i,y)$ element of $V$ is $r'_i(y)$.
Hence, $\mathrm{rank}_{\mathbb{F}_2}(M_{f'})\le t$.
\end{proof}
\begin{remark}
If we restrict the decomposition to be symmetric, i.e., $l_i = r_i$ for all $i=1,\dotsc,t$, extra 1 dimension is required for arbitrary XOR function $g^\oplus$~\cite{doi:10.1137/0204014}.
\end{remark}

\begin{lemma}\label{lem:nlbc1}
For any $g\colon\{0,1\}^n\to\{0,1\}$, $\mathrm{NLBC}(g^\oplus)=0$ only when $g$ is a parity of some variables or its negation.
Furthermore, $\mathrm{NLBC}(g^\oplus)$ cannot be equal to 1.
\end{lemma}
\begin{proof}
From Theorem~\ref{thm:f2rank}, $\mathrm{NLBC}(g^\oplus)=0$ implies
$g(x\oplus y)\oplus g(x) \oplus g(y) \oplus g(0)=0$.
Hence, it holds $g(x\oplus y)\oplus g(0) = (g(x) \oplus g(0)) \oplus (g(y) \oplus g(0))$, so that $g(z)\oplus g(0)$ is linear, i.e., parity of some variables.
Assume $\mathrm{NLBC}(g^\oplus)=1$. From Theorem~\ref{thm:f2rank}, $\mathrm{rank}_{\mathbb{F}_2}(M_{{g^{\oplus}}'})$ must be equal to 1.
Since $M_{{g^\oplus}'}$ is a symmetric matrix, there is a decomposition $M_{{g^\oplus}'}=vv^t$ where $v$ denotes a $\mathbb{F}_2$-vector of length $2^n$.
On the other hand, the diagonal elements of $M_{{g^\oplus}'}$ must be zero.
That implies $v=0$, and hence $\mathrm{NLBC}(g^\oplus)=0$.
This is a contradiction.
\end{proof}

\begin{example}\label{exm:majn}
The following table shows the nonlocal box complexity of $\mathrm{Maj}_n$ computed numerically by a computer.
\begin{ruledtabular}
\begin{tabular}{ccccccccc}
$n$ & 3 & 5 & 7 & 9 & 11 & 13 & 15 & 17\\ \hline
$\mathrm{NLBC}(\mathrm{Maj}_n^\oplus)$ & 2 & 14 & 26 & 254 & 494 & 1090 & 1818 & 65534\\
\end{tabular}
\end{ruledtabular}
\end{example}
In Example~\ref{exm:majn}, it is not easy to find any rule between $n$ and the nonlocal box complexity although $\mathrm{NLBC}(\mathrm{Maj}_n^\oplus)=2^{n-1}-2$
may happen frequently, e.g., $n=3,5,9,17$.
Generally, it is considered to be difficult to express $\mathrm{rank}_{\mathbb{F}_2}(M_f)$ in a simple form for arbitrary given $f$.
Note that the rank on $\mathbb{R}$ is always at least the rank on $\mathbb{F}_2$.
Since $\mathrm{rank}_{\mathbb{R}}(M_{g^\oplus})$ is equal to the number of nonzero Fourier coefficients of $g$~\cite{10.1109/12.755000},
$2^{n-1}+1$ is an upper bound of $\mathrm{NLBC}(\mathrm{Maj}_n^\oplus)$ for odd $n$ (an inequality $\mathrm{rank}_{\mathbb{F}_2}(M_f)-2\le \mathrm{NLBC}(f)\le \mathrm{rank}_{\mathbb{F}_2}(M_f)$ can be obtained in a similar way as Theorem~\ref{thm:f2rank}).
Here, we introduce a lower bound of the nonlocal box complexity using the one-way communication complexity.

\begin{lemma}\label{lem:xor-oneway}
For any $f\colon\{0,1\}^n\times\{0,1\}^n\to\{0,1\}$,
\begin{equation*}
\mathrm{NLBC}(f) \ge \max\left\{D_\rightarrow^\oplus(f), D_\leftarrow^\oplus(f)\right\}.
\end{equation*}
\end{lemma}
\begin{proof}
Assume $f(x,y)$ has the form~\eqref{eq:F2expand}.
Bob can compute $f\oplus A(x)$ from $(l_i(x))_{i=1,\dotsc,\mathrm{NLBC}(f)}$.
\end{proof}

It obviously holds $D_\rightarrow^\oplus(f)\ge D_\rightarrow(f)-1$.
If $g$ is an odd function, i.e., $g(\overline{z})=\overline{g(z)}$ where $\overline{z}$ denotes the bit inversion of $z$, then
$D_\rightarrow^\oplus(g^\oplus)=D_\rightarrow(g^\oplus)-1$
since $g(x\oplus y)\oplus g(x)= g(\overline{x}\oplus y)\oplus g(\overline{x})$.

\begin{example}\label{exm:majlow}
It obviously holds $D_\rightarrow(\mathrm{Maj}_n^\oplus)=n$.
Since $\mathrm{Maj}_n$ is an odd function, it holds $D_\rightarrow^\oplus(\mathrm{Maj}_n^\oplus)=n-1$.
From Example~\ref{exm:majn}, this lower bound is tight for $n=3$, but becomes looser as $n$ increases.
This lower bound seems not to be asymptotically tight.
\end{example}

In fact, the adaptive protocol introduced in the next section has bias $\delta^{D_\rightarrow^\oplus(g^\oplus)}$ for arbitrary XOR function $g^\oplus$.

\section{Adaptive protocol}\label{sec:adaptive}
\subsection{Paw\l owski et al.'s protocol}
In this section, we show a new adaptive protocol which is inspired by the adaptive protocol invented in~\cite{pawlowski2009information}.
Let the address function $\mathrm{Addr}_n$ be
\begin{equation*}
\mathrm{Addr}_n(x_0,\dotsc,x_{2^n-1},y_1,\dotsc,y_n) := x_y
\end{equation*}
where $y=\sum_{i=1}^{n}y_i2^{i-1}$.
In~\cite{pawlowski2009information}, Paw\l owski et al.\@ characterized the quantum limit $\frac{2+\sqrt{2}}{4}$ of the CHSH probability by using a new principle called information causality.
What they essentially showed in~\cite{pawlowski2009information} is following.

\begin{lemma}\label{lem:addr}
There is a PR-correct protocol computing the address function $\mathrm{Addr}_n$ with bias $\delta^{n}$.
\end{lemma}
\begin{proof}
The lemma is shown by the induction.
There is an representation
\begin{equation*}
\mathrm{Addr}_1(x_0,x_1,y_1) = x_0 \oplus y_1 (x_0\oplus x_1).
\end{equation*}
Hence, there exists a non-adaptive protocol computing $\mathrm{Addr}_1$ with bias $\delta$, so that the lemma holds for $n=1$.
For $n\ge 2$, there is a recursive formula
\begin{equation*}
\mathrm{Addr}_n(x_0,\dotsc,x_{2^n-1},y_1,\dotsc,y_n)
=
\mathrm{Addr}_1(x'_0, x'_1, y_n)
\end{equation*}
where
\begin{align*}
x'_0 &:= \mathrm{Addr}_{n-1}(x_0,\dotsc,x_{2^{n-1}-1}, y_1,\dotsc,y_{n-1})\\
x'_1 &:= \mathrm{Addr}_{n-1}(x_{2^{n-1}},\dotsc,x_{2^{n}-1}, y_1,\dotsc,y_{n-1}).
\end{align*}
From the hypothesis of the induction, there is a PR-correct protocol computing $x'_0$ and $x'_1$ with bias $\delta^{n-1}$.
Let $a_0$ and $b_0$ ($a_1$ and $b_1$) be random variables corresponding to the outputs of the protocol computing $x'_0$ ($x'_1$), respectively.
Then, if $\delta=1$, one obtains
\begin{align*}
&\mathrm{Addr}_n(x_0,\dotsc,x_{2^n-1},y_1,\dotsc,y_n)\\
&=\mathrm{Addr}_1(a_0\oplus b_0, a_1\oplus b_1, y_n)\\
&=\mathrm{Addr}_1(a_0, a_1, y_n)\oplus \mathrm{Addr}_1(b_0, b_1, y_n)\\
&=a_0 \oplus y_n (a_0\oplus a_1) \oplus b_{y_n}.
\end{align*}
From this observation, we recursively define the protocol for $\mathrm{Addr}_n$ in the following way.
(P1) Compute $a_0$ and $a_1$ at Alice's side, and $b_{y_n}$ at Bob's side using the protocol for $\mathrm{Addr}_{n-1}$.
(P2) Input $a_0\oplus a_1$ and $y_n$ into the common nonlocal box, and obtain $a'$ and $b'$.
(P3) Output $a:=a_0\oplus a'$ at Alice's side and $b:=b'\oplus b_{y_n}$ at Bob's side.
This protocol is obviously PR-correct.
Since at each step, the error of bias $\delta$ is XORed, this protocol has bias $\delta^n$.
\end{proof}

\subsection{The adaptive protocol}\label{subsec:adaptive}
In the following, we show a new adaptive protocol computing arbitrary given function $f\colon\{0,1\}^n\times\{0,1\}^n\to\{0,1\}$ using Paw\l owski et al's protocol.
\begin{theorem}\label{thm:adaptive}
For arbitrary function $f\colon\{0,1\}^n\times\{0,1\}^n\to\{0,1\}$,
there is a PR-correct protocol computing $f$ with bias $\delta^{\min\{D_\rightarrow^\oplus(f), D_\leftarrow^\oplus(f)\}}$.
\end{theorem}
\begin{proof}
Arbitrary function $f$ can be represented by
\begin{align*}
f(x,y) &= \mathrm{Addr}_{n}\bigl(f(x,0,\dotsc,0),f(x,0\dotsc,0,1),\\
&\quad\dotsc,f(x,1,\dotsc,1),y_1,\dotsc,y_n\bigr).
\end{align*}
From Lemma~\ref{lem:addr}, there is an adaptive protocol computing $f$ with bias $\delta^n$.

We can consider compression of Bob's input since we do not have to distinguish $y$'s belong to equivalent columns of $M_f$.
By applying the compression, we obtain the protocol with bias $\delta^{D_\leftarrow(f)}$.
Furthermore, if we have an XOR protocol for $f(x,y)\oplus B(y)$, we also obtain an XOR protocol for $f(x,y)$
by replacing Bob's output $b$ with $b\oplus B(y)$.
Hence, we obtain the protocol with bias $\delta^{D_\leftarrow^\oplus(f)}$.
In the same way, we also obtain the protocol with bias $\delta^{D_\rightarrow^\oplus(f)}$.
\end{proof}

From Lemma~\ref{lem:xor-oneway} and Theorem~\ref{thm:adaptive}, we obtain the following corollary.
\begin{corollary}
The adaptive PR-correct protocol in Theorem~\ref{thm:adaptive} is no worse than any non-adaptive PR-correct protocol.
\end{corollary}

\section{Bias amplification}\label{sec:ba}
We now consider the bias amplification by general XOR function $g^\oplus$ in Brassard et al.'s protocol where $g^\oplus$ is computed by the adaptive PR-correct protocol introduced in Theorem~\ref{thm:adaptive}.
If $z$ is a random variable taking $+1$ with probability $\frac{1+\epsilon}2$ and $-1$ with probability $\frac{1-\epsilon}2$,
its expectation is $\epsilon$.
The expectation $\epsilon$ is called the bias of random variable $z$.
If the inputs for $g$ is independently and identically distributed and have bias $\epsilon$, the bias of output of $g$ is given in the following formula.

\begin{definition}
For any $g\colon\{0,1\}^n\to\{0,1\}$,
we define 
\begin{equation*}
\mathrm{Bias}_\epsilon(g) := 
\sum_{S\subseteq[n]} \widehat{g}(S) \epsilon^{|S|}.
\end{equation*}
\end{definition}

\begin{example}
Since $\mathrm{Maj}_3(z_1,z_2,z_3)=(1/2)(z_1+z_2+z_3-z_1z_2z_3)$, one obtains
$\mathrm{Bias}_\epsilon({\mathrm{Maj}_3}) = (3/2) \epsilon - (1/2) \epsilon^3$.
Roughly speaking, the input bias $\epsilon$ is amplified to $(3/2)\epsilon$ for small $\epsilon$.
\end{example}

When a boolean function $g$ is computed correctly with probability $\frac{1+\rho}2$, the output bias of $g$ is $\rho \mathrm{Bias}_\epsilon(g)$.
We say that the bias is amplified by $g$ if the absolute value of bias of output of $g$ is larger than that of input and if the sign of bias is preserved.
The bias is amplified by the noisy $g$ for sufficiently small input bias iff $\mathrm{Bias}_0(g)=0$ and $\rho \left.\frac{\mathrm{d} \mathrm{Bias}_\epsilon(g)}{\mathrm{d}\epsilon}\right|_{\epsilon=0}>1$.
Hence, we obtain the following theorem.

\begin{theorem}
Assume that $g\colon\{0,1\}^n\to\{0,1\}$
 can be computed correctly with probability $\frac{1+\rho}2$.
Then, the bias is amplified by the noisy $g$ when the input bias is sufficiently small iff
 $\widehat{g}(\varnothing)=0$ and $\rho>\rho_{\mathrm{B}}(g)$ where
\begin{equation*}
\rho_{\mathrm{B}}(g) := \frac1{\max\left\{1,\,\sum_{i=1}^n\widehat{g}(\{i\})\right\}}.
\end{equation*}
\end{theorem}

The majority functions minimize $\rho_\mathrm{B}(g)$.
\begin{lemma}\label{lem:majmax}
For $g\colon\{0,1\}^n\to\{0,1\}$,
\begin{align*}
\rho_{\mathrm{B}}(g) &\ge \frac{2^{n-1}}{n\binom{n-1}{\frac{n-1}2}},&&\text{if $n$ is odd}\\
\rho_{\mathrm{B}}(g) &\ge \frac{2^{n}}{n\binom{n}{\frac{n}2}},&&\text{if $n$ is even.}
\end{align*}
The equality is achieved by and only by the majority functions on $n$ variables.
Asymptotically, it holds $\rho_{\mathrm{B}}(g) \ge \sqrt{\pi/(2n)}(1+O(n^{-1/2}))$.
\end{lemma}
\begin{proof}
One obtains $\sum_{i\in[n]} \widehat{g}(\{i\})=\mathbb{E}[g(x)(x_1+\dotsb+x_n)]\le\mathbb{E}[|x_1+\dotsb+x_n|]$
where the equality holds only when $g$ is $\mathrm{Maj}_n$~\cite{odonnell2014analysis}.
Hence, only the majority functions $\mathrm{Maj}_n$ maximize $\sum_{i\in[n]}\widehat{g}(\{i\})$.
It is easy to complete the rest of the proof~\cite{odonnell2014analysis}.
\end{proof}
Note that the lower bound for even $n$ is equal to the lower bound for $n-1$.
The condition on $\delta$ for the bias amplification by Brassard et al.'s protocol is $\delta^{D_{\rightarrow}^\oplus(g^\oplus)}> \rho_{\mathrm{B}}(g)$.

\begin{definition}
For any $g\colon \{0,1\}^n\to\{0,1\}$,
\begin{equation*}
\delta_\mathrm{B}(g) :=
\begin{cases}
\rho_{\mathrm{B}}(g)^{\frac1{D_\rightarrow^\oplus(g^\oplus)}},& \text{if $\widehat{g}(\varnothing)=0$ and $\rho_{\mathrm{B}}(g)<1$}\\
1,&\text{otherwise.}
\end{cases}
\end{equation*}
\end{definition}

If $\delta > \delta_{\mathrm{B}}(g)$ for some $g\colon\{0,1\}^n\to\{0,1\}$, there exists an XOR protocol with constant bias.

\begin{example}\label{exm:maj3brassard}
One obtains $\delta_\mathrm{B}(\mathrm{Maj}_3) = \sqrt{2/3}$
that means that the threshold for the CHSH probability is $\frac{1+\sqrt{2/3}}2 = \frac{3+\sqrt{6}}6$~\cite{PhysRevLett.96.250401}.
\end{example}

We can now rephrase Theorem~\ref{thm:main} in the following form.
\begin{theorem}\label{thm:main1}
\begin{equation*}
\inf_{g\colon \{0,1\}^n\to\{0,1\}, n\in\mathbb{N}} \delta_{\mathrm{B}}(g) = \sqrt{\frac23}.
\end{equation*}
Furthermore, $\delta_{\mathrm{B}}(g) = \sqrt{2/3}$ iff $g$ is essentially equivalent to $\mathrm{Maj}_3$.
\end{theorem}

Here, we say that $g$ is essentially equivalent to $\mathrm{Maj}_3$ if $g$ is the majority of some fixed three input variables and ignores the other $n-3$ input variables.
The following lemma was shown in~\cite{montanaro2009ccxor}.
\begin{lemma}\label{lem:fdim}
For any $g\colon\{0,1\}^n\to\{0,1\}$,
\begin{equation*}
D_\rightarrow(g^\oplus)=\dim(\widehat{g}).
\end{equation*}
\end{lemma}
Since $D_\rightarrow^\oplus(f)\ge D_\rightarrow(f)-1$, it holds $D_\rightarrow^\oplus(g^\oplus)\ge \dim(\widehat{g})-1$.
\begin{remark}
If $A(x)$ in the definition of $D_\rightarrow^\oplus(f)$ is restricted to be linear, $D_\rightarrow^\oplus(g^\oplus)$ is equal to the affine dimension of $\widehat{g}$, which is
the minimum dimension of affine space on $\mathbb{F}_2$ including $\{1_S \mid S\in\mathrm{supp}(\widehat{g})\}$.
Hence, the affine dimension of $\widehat{g}$ is an upper bound of $D_\rightarrow^\oplus(g^\oplus)$.
\end{remark}
First, we show that Theorem~\ref{thm:main1} holds for $n\le 4$.
\begin{lemma}\label{lem:exh}
It holds $\delta_{\mathrm{B}}(g) \ge \sqrt{2/3}$ for all boolean functions $g$ on at most 4 variables.
Furthermore, for $n\le 4$, only functions essentially equivalent to $\mathrm{Maj}_3$ satisfy $\delta_{\mathrm{B}}(g)=\sqrt{2/3}$.
\end{lemma}
\begin{proof}
Assume $D_\rightarrow^\oplus(g^\oplus)\le 1$.
Then, the protocol is non-adaptive.
From Lemma~\ref{lem:nlbc1}, $g$ must be linear, and hence, $\rho_\mathrm{B}(g)=1$.
Assume $D_\rightarrow^\oplus(g^\oplus)\ge 2$.
From Lemma~\ref{lem:majmax}, $\rho_\mathrm{B}(g)\ge 2/3$ for $n\le 4$, and hence, $\delta_{\mathrm{B}}(g)\ge\sqrt{2/3}$.
From Example~\ref{exm:maj3brassard}, it is achieved by $\mathrm{Maj}_3$.

Next, we show the uniqueness.
From the above argument, it holds $\delta_\mathrm{B}(g)=\sqrt{2/3}$ only when
$\rho_\mathrm{B}(g)=2/3$ and $D_\rightarrow^\oplus(g^\oplus)=2$.
From Lemma~\ref{lem:majmax}, $\rho_{\mathrm{B}}(g)=2/3$ only when $g$ is one of the 64 majority functions on 4 variables.
In the following, we show that for $g\in\mathrm{Maj}_4$, $D_\rightarrow^\oplus(g^\oplus)=2$ only when $g$ is essentially equivalent to $\mathrm{Maj}_3$.
From Lemma~\ref{lem:fdim}, $|\{i\in[n]\mid\widehat{g}(\{i\})\ne 0\}|\le \dim(\widehat{g})\le 3$.
If $|\{i\in[n]\mid\widehat{g}(\{i\})\ne 0\}|\le 2$,
it holds $\sum_{i\in[n]}\widehat{g}(\{i\})\le \sqrt{2} < 3/2$ from the Cauchy-Schwartz inequality.
If $|\{i\in[n]\mid\widehat{g}(\{i\})\ne 0\}|= 3$, $g$ depends only on three variables since $\dim(\widehat{g})\le 3$.
Hence, $g$ is essentially equivalent to $\mathrm{Maj}_3$.
\end{proof}
From the following lemma, only boolean functions with small Fourier dimension may outperform $\mathrm{Maj}_3$.

\begin{lemma}\label{lem:sum}
For any $g\colon\{0,1\}^n\to\{0,1\}$,
\begin{equation*}
\delta_{\mathrm{B}}(g)
\ge \left(\frac1{\dim(\widehat{g})}\right)^{\frac1{2(\dim(\widehat{g})-1)}}.
\end{equation*}
In particular, if $\dim(\widehat{g})\ge 5$, it holds $\delta_{\mathrm{B}}(g)>\sqrt{2/3}$.
\end{lemma}
\begin{proof}
One obtains
\begin{align*}
\dim(\widehat{g})&\ge |\{i\in[n] \mid \widehat{g}(\{i\})\ne 0\}|\\
&\ge 
\frac{\left(\sum_{i\in[n]}\widehat{g}(\{i\})\right)^2}{\sum_{i\in[n]}\widehat{g}(\{i\})^2}
\ge 
\left(\sum_{i\in[n]}\widehat{g}(\{i\})\right)^2.
\end{align*}
In the above, the first inequality is trivial. The second inequality is the Cauchy-Schwartz inequality.
The third inequality holds since sum of squares of all of the Fourier coefficients is 1.
Hence, $\rho_{\mathrm{B}}(g) \ge \dim(\widehat{g})^{-1/2}$.
From Lemma~\ref{lem:fdim}, we obtain this theorem.
\end{proof}

Lemmas~\ref{lem:exh} and~\ref{lem:sum} give the complete proof of Theorem~\ref{thm:main1}.
\begin{proof}[Proof of Theorem~\ref{thm:main1}]
From Lemma~\ref{lem:sum},
we only have to show that 
if $|\{i\in[n] \mid \widehat{g}(\{i\})\ne 0\}|\le \dim(\widehat{g}) \le 4$,
$\delta_{\mathrm{B}}(g) \le \sqrt{2/3}$ only for $g$ essentially equivalent to $\mathrm{Maj}_3$.
Assume $|\{i\in[n] \mid \widehat{g}(\{i\})\ne 0\}|= 4$.
Then, the boolean function $g$ depends only on 4 input variables since $\dim(\widehat{g})\le 4$.
From Lemma~\ref{lem:exh}, there is no function on 4 variables satisfying $\delta_{\mathrm{B}}(g)\le \sqrt{2/3}$ except for functions essentially equivalent to $\mathrm{Maj}_3$.
Next, we assume $|\{i\in[n] \mid \widehat{g}(\{i\})\ne 0\}|= 3$.
In this case, $\sum_{i\in[n]}\widehat{g}(\{i\})\le \sqrt{3}$.
Since $(1/\sqrt{3})^{1/3}> \sqrt{2/3}$, we can assume $D_\rightarrow^\oplus(g^\oplus)\le 2$.
Then, the boolean function $g$ depends only on 3 input variables since $\dim(\widehat{g})\le D_\rightarrow^\oplus(g^\oplus)+1 \le 3$.
From Lemma~\ref{lem:exh}, there is no function on 3 variables satisfying $\delta_{\mathrm{B}}(g)\le\sqrt{2/3}$ except for $\mathrm{Maj}_3$.
Next, we assume $|\{i\in[n] \mid \widehat{g}(\{i\})\ne 0\}|\le 2$.
In this case, $\sum_{i\in[n]}\widehat{g}(\{i\})\le \sqrt{2}$.
Since $(1/\sqrt{2})^{1/2}> \sqrt{2/3}$, we can assume $D_\rightarrow^\oplus(g^\oplus)\le 1$.
From Lemma~\ref{lem:nlbc1}, it holds $\delta_{\mathrm{B}}(g)= 1$.
We conclude that there is no function satisfying $\delta_{\mathrm{B}}(g)\le \sqrt{2/3}$ except for functions essentially equivalent to $\mathrm{Maj}_3$.
\end{proof}

\section{Conclusion}
In this paper, we show that the 3-input majority function is the unique optimal function for Brassard et al.'s bias amplification on some conditions.
This paper also develops mathematical framework using Fourier analysis for problems on XOR protocols with nonlocal boxes.
On the other hand, in this paper, functions $g^\oplus$ for the bias amplification are restricted to be XOR function although it seems to be a natural restriction.
Furthermore, protocols for computing functions $g^\oplus$, in this paper, are restricted to be the particular adaptive PR-correct protocol, which is better than arbitrary non-adaptive PR-correct protocol.
General adaptive protocols may allow more reliable computation than these protocols~\cite{PhysRevLett.102.160403}.
Similar adaptive protocol in Section~\ref{subsec:adaptive} gives the bias $\delta^{D_{\leftrightarrow}^\oplus(f)}$ where $D_{\leftrightarrow}^\oplus(f)$ denotes the two-way communication complexity for computing $a$ and $b$ satisfying $a\oplus b=f(x,y)$.
Hence, the result of this paper does not show the limitation of the idea of the bias amplification, but show only the limitation of the idea of the bias amplification by XOR function computed by the particular adaptive PR-correct protocol.
The bias amplification by general adaptive computation of non-XOR function would be an interesting direction of research.

\begin{acknowledgments}
This work was supported by MEXT KAKENHI Grant Number 24106008.
\end{acknowledgments}


\begin{thebibliography}{19}%
\makeatletter
\providecommand \@ifxundefined [1]{%
 \@ifx{#1\undefined}
}%
\providecommand \@ifnum [1]{%
 \ifnum #1\expandafter \@firstoftwo
 \else \expandafter \@secondoftwo
 \fi
}%
\providecommand \@ifx [1]{%
 \ifx #1\expandafter \@firstoftwo
 \else \expandafter \@secondoftwo
 \fi
}%
\providecommand \natexlab [1]{#1}%
\providecommand \enquote  [1]{``#1''}%
\providecommand \bibnamefont  [1]{#1}%
\providecommand \bibfnamefont [1]{#1}%
\providecommand \citenamefont [1]{#1}%
\providecommand \href@noop [0]{\@secondoftwo}%
\providecommand \href [0]{\begingroup \@sanitize@url \@href}%
\providecommand \@href[1]{\@@startlink{#1}\@@href}%
\providecommand \@@href[1]{\endgroup#1\@@endlink}%
\providecommand \@sanitize@url [0]{\catcode `\\12\catcode `\$12\catcode
  `\&12\catcode `\#12\catcode `\^12\catcode `\_12\catcode `\%12\relax}%
\providecommand \@@startlink[1]{}%
\providecommand \@@endlink[0]{}%
\providecommand \url  [0]{\begingroup\@sanitize@url \@url }%
\providecommand \@url [1]{\endgroup\@href {#1}{\urlprefix }}%
\providecommand \urlprefix  [0]{URL }%
\providecommand \Eprint [0]{\href }%
\providecommand \doibase [0]{http://dx.doi.org/}%
\providecommand \selectlanguage [0]{\@gobble}%
\providecommand \bibinfo  [0]{\@secondoftwo}%
\providecommand \bibfield  [0]{\@secondoftwo}%
\providecommand \translation [1]{[#1]}%
\providecommand \BibitemOpen [0]{}%
\providecommand \bibitemStop [0]{}%
\providecommand \bibitemNoStop [0]{.\EOS\space}%
\providecommand \EOS [0]{\spacefactor3000\relax}%
\providecommand \BibitemShut  [1]{\csname bibitem#1\endcsname}%
\let\auto@bib@innerbib\@empty
\bibitem [{\citenamefont {Bell}(1964)}]{bell1964einstein}%
  \BibitemOpen
  \bibfield  {author} {\bibinfo {author} {\bibfnamefont {J.~S.}\ \bibnamefont
  {Bell}},\ }\href@noop {} {\bibfield  {journal} {\bibinfo  {journal}
  {Physics}\ }\textbf {\bibinfo {volume} {1}},\ \bibinfo {pages} {195}
  (\bibinfo {year} {1964})}\BibitemShut {NoStop}%
\bibitem [{\citenamefont {Clauser}\ \emph {et~al.}(1969)\citenamefont
  {Clauser}, \citenamefont {Horne}, \citenamefont {Shimony},\ and\
  \citenamefont {Holt}}]{PhysRevLett.23.880}%
  \BibitemOpen
  \bibfield  {author} {\bibinfo {author} {\bibfnamefont {J.~F.}\ \bibnamefont
  {Clauser}}, \bibinfo {author} {\bibfnamefont {M.~A.}\ \bibnamefont {Horne}},
  \bibinfo {author} {\bibfnamefont {A.}~\bibnamefont {Shimony}}, \ and\
  \bibinfo {author} {\bibfnamefont {R.~A.}\ \bibnamefont {Holt}},\ }\href
  {\doibase 10.1103/PhysRevLett.23.880} {\bibfield  {journal} {\bibinfo
  {journal} {Phys. Rev. Lett.}\ }\textbf {\bibinfo {volume} {23}},\ \bibinfo
  {pages} {880} (\bibinfo {year} {1969})}\BibitemShut {NoStop}%
\bibitem [{\citenamefont {Fine}(1982)}]{PhysRevLett.48.291}%
  \BibitemOpen
  \bibfield  {author} {\bibinfo {author} {\bibfnamefont {A.}~\bibnamefont
  {Fine}},\ }\href {\doibase 10.1103/PhysRevLett.48.291} {\bibfield  {journal}
  {\bibinfo  {journal} {Phys. Rev. Lett.}\ }\textbf {\bibinfo {volume} {48}},\
  \bibinfo {pages} {291} (\bibinfo {year} {1982})}\BibitemShut {NoStop}%
\bibitem [{\citenamefont {Cirel'son}(1980)}]{cirelson1980}%
  \BibitemOpen
  \bibfield  {author} {\bibinfo {author} {\bibfnamefont {B.~S.}\ \bibnamefont
  {Cirel'son}},\ }\href {\doibase 10.1007/BF00417500} {\bibfield  {journal}
  {\bibinfo  {journal} {Lett. Math. Phys.}\ }\textbf {\bibinfo {volume} {4}},\
  \bibinfo {pages} {93} (\bibinfo {year} {1980})}\BibitemShut {NoStop}%
\bibitem [{\citenamefont {Popescu}\ and\ \citenamefont
  {Rohrlich}(1994)}]{PR1994}%
  \BibitemOpen
  \bibfield  {author} {\bibinfo {author} {\bibfnamefont {S.}~\bibnamefont
  {Popescu}}\ and\ \bibinfo {author} {\bibfnamefont {D.}~\bibnamefont
  {Rohrlich}},\ }\href {\doibase 10.1007/BF02058098} {\bibfield  {journal}
  {\bibinfo  {journal} {Found. Phyis.}\ }\textbf {\bibinfo {volume} {24}},\
  \bibinfo {pages} {379} (\bibinfo {year} {1994})}\BibitemShut {NoStop}%
\bibitem [{\citenamefont {van Dam}(2005)}]{vandam2005}%
  \BibitemOpen
  \bibfield  {author} {\bibinfo {author} {\bibfnamefont {W.}~\bibnamefont {van
  Dam}},\ }\href@noop {} {\enquote {\bibinfo {title} {Implausible consequences
  of superstrong nonlocality},}\ } (\bibinfo {year} {2005}),\ \Eprint
  {http://arxiv.org/abs/quant-ph/0501159} {arXiv:quant-ph/0501159} \BibitemShut
  {NoStop}%
\bibitem [{\citenamefont {Brassard}\ \emph {et~al.}(2006)\citenamefont
  {Brassard}, \citenamefont {Buhrman}, \citenamefont {Linden}, \citenamefont
  {M\'ethot}, \citenamefont {Tapp},\ and\ \citenamefont
  {Unger}}]{PhysRevLett.96.250401}%
  \BibitemOpen
  \bibfield  {author} {\bibinfo {author} {\bibfnamefont {G.}~\bibnamefont
  {Brassard}}, \bibinfo {author} {\bibfnamefont {H.}~\bibnamefont {Buhrman}},
  \bibinfo {author} {\bibfnamefont {N.}~\bibnamefont {Linden}}, \bibinfo
  {author} {\bibfnamefont {A.~A.}\ \bibnamefont {M\'ethot}}, \bibinfo {author}
  {\bibfnamefont {A.}~\bibnamefont {Tapp}}, \ and\ \bibinfo {author}
  {\bibfnamefont {F.}~\bibnamefont {Unger}},\ }\href {\doibase
  10.1103/PhysRevLett.96.250401} {\bibfield  {journal} {\bibinfo  {journal}
  {Phys. Rev. Lett.}\ }\textbf {\bibinfo {volume} {96}},\ \bibinfo {pages}
  {250401} (\bibinfo {year} {2006})}\BibitemShut {NoStop}%
\bibitem [{\citenamefont {Paw{\l}owski}\ \emph {et~al.}(2009)\citenamefont
  {Paw{\l}owski}, \citenamefont {Paterek}, \citenamefont {Kaszlikowski},
  \citenamefont {Scarani}, \citenamefont {Winter},\ and\ \citenamefont
  {{\.Z}ukowski}}]{pawlowski2009information}%
  \BibitemOpen
  \bibfield  {author} {\bibinfo {author} {\bibfnamefont {M.}~\bibnamefont
  {Paw{\l}owski}}, \bibinfo {author} {\bibfnamefont {T.}~\bibnamefont
  {Paterek}}, \bibinfo {author} {\bibfnamefont {D.}~\bibnamefont
  {Kaszlikowski}}, \bibinfo {author} {\bibfnamefont {V.}~\bibnamefont
  {Scarani}}, \bibinfo {author} {\bibfnamefont {A.}~\bibnamefont {Winter}}, \
  and\ \bibinfo {author} {\bibfnamefont {M.}~\bibnamefont {{\.Z}ukowski}},\
  }\href {\doibase 10.1038/nature08400} {\bibfield  {journal} {\bibinfo
  {journal} {Nature}\ }\textbf {\bibinfo {volume} {461}},\ \bibinfo {pages}
  {1101} (\bibinfo {year} {2009})}\BibitemShut {NoStop}%
\bibitem [{\citenamefont {O'Donnell}(2014)}]{odonnell2014analysis}%
  \BibitemOpen
  \bibfield  {author} {\bibinfo {author} {\bibfnamefont {R.}~\bibnamefont
  {O'Donnell}},\ }\href@noop {} {\emph {\bibinfo {title} {Analysis of Boolean
  Functions}}}\ (\bibinfo  {publisher} {Cambridge University Press},\ \bibinfo
  {year} {2014})\BibitemShut {NoStop}%
\bibitem [{\citenamefont {Masanes}\ \emph {et~al.}(2006)\citenamefont
  {Masanes}, \citenamefont {Acin},\ and\ \citenamefont
  {Gisin}}]{PhysRevA.73.012112}%
  \BibitemOpen
  \bibfield  {author} {\bibinfo {author} {\bibfnamefont {L.}~\bibnamefont
  {Masanes}}, \bibinfo {author} {\bibfnamefont {A.}~\bibnamefont {Acin}}, \
  and\ \bibinfo {author} {\bibfnamefont {N.}~\bibnamefont {Gisin}},\ }\href
  {\doibase 10.1103/PhysRevA.73.012112} {\bibfield  {journal} {\bibinfo
  {journal} {Phys. Rev. A}\ }\textbf {\bibinfo {volume} {73}},\ \bibinfo
  {pages} {012112} (\bibinfo {year} {2006})}\BibitemShut {NoStop}%
\bibitem [{\citenamefont {Short}(2009)}]{PhysRevLett.102.180502}%
  \BibitemOpen
  \bibfield  {author} {\bibinfo {author} {\bibfnamefont {A.~J.}\ \bibnamefont
  {Short}},\ }\href {\doibase 10.1103/PhysRevLett.102.180502} {\bibfield
  {journal} {\bibinfo  {journal} {Phys. Rev. Lett.}\ }\textbf {\bibinfo
  {volume} {102}},\ \bibinfo {pages} {180502} (\bibinfo {year}
  {2009})}\BibitemShut {NoStop}%
\bibitem [{\citenamefont {Forster}\ \emph {et~al.}(2009)\citenamefont
  {Forster}, \citenamefont {Winkler},\ and\ \citenamefont
  {Wolf}}]{PhysRevLett.102.120401}%
  \BibitemOpen
  \bibfield  {author} {\bibinfo {author} {\bibfnamefont {M.}~\bibnamefont
  {Forster}}, \bibinfo {author} {\bibfnamefont {S.}~\bibnamefont {Winkler}}, \
  and\ \bibinfo {author} {\bibfnamefont {S.}~\bibnamefont {Wolf}},\ }\href
  {\doibase 10.1103/PhysRevLett.102.120401} {\bibfield  {journal} {\bibinfo
  {journal} {Phys. Rev. Lett.}\ }\textbf {\bibinfo {volume} {102}},\ \bibinfo
  {pages} {120401} (\bibinfo {year} {2009})}\BibitemShut {NoStop}%
\bibitem [{\citenamefont {Brunner}\ and\ \citenamefont
  {Skrzypczyk}(2009)}]{PhysRevLett.102.160403}%
  \BibitemOpen
  \bibfield  {author} {\bibinfo {author} {\bibfnamefont {N.}~\bibnamefont
  {Brunner}}\ and\ \bibinfo {author} {\bibfnamefont {P.}~\bibnamefont
  {Skrzypczyk}},\ }\href {\doibase 10.1103/PhysRevLett.102.160403} {\bibfield
  {journal} {\bibinfo  {journal} {Phys. Rev. Lett.}\ }\textbf {\bibinfo
  {volume} {102}},\ \bibinfo {pages} {160403} (\bibinfo {year}
  {2009})}\BibitemShut {NoStop}%
\bibitem [{\citenamefont {von Neumann}(1956)}]{neumann1956}%
  \BibitemOpen
  \bibfield  {author} {\bibinfo {author} {\bibfnamefont {J.}~\bibnamefont {von
  Neumann}},\ }\href@noop {} {\bibfield  {journal} {\bibinfo  {journal}
  {Automata Studies}\ ,\ \bibinfo {pages} {43}} (\bibinfo {year}
  {1956})}\BibitemShut {NoStop}%
\bibitem [{\citenamefont {Linden}\ \emph {et~al.}(2007)\citenamefont {Linden},
  \citenamefont {Popescu}, \citenamefont {Short},\ and\ \citenamefont
  {Winter}}]{PhysRevLett.99.180502}%
  \BibitemOpen
  \bibfield  {author} {\bibinfo {author} {\bibfnamefont {N.}~\bibnamefont
  {Linden}}, \bibinfo {author} {\bibfnamefont {S.}~\bibnamefont {Popescu}},
  \bibinfo {author} {\bibfnamefont {A.~J.}\ \bibnamefont {Short}}, \ and\
  \bibinfo {author} {\bibfnamefont {A.}~\bibnamefont {Winter}},\ }\href
  {\doibase 10.1103/PhysRevLett.99.180502} {\bibfield  {journal} {\bibinfo
  {journal} {Phys. Rev. Lett.}\ }\textbf {\bibinfo {volume} {99}},\ \bibinfo
  {pages} {180502} (\bibinfo {year} {2007})}\BibitemShut {NoStop}%
\bibitem [{\citenamefont {Kaplan}\ \emph {et~al.}(2011)\citenamefont {Kaplan},
  \citenamefont {Laplante}, \citenamefont {Kerenidis},\ and\ \citenamefont
  {Roland}}]{kaplan2011nonlocal}%
  \BibitemOpen
  \bibfield  {author} {\bibinfo {author} {\bibfnamefont {M.}~\bibnamefont
  {Kaplan}}, \bibinfo {author} {\bibfnamefont {S.}~\bibnamefont {Laplante}},
  \bibinfo {author} {\bibfnamefont {I.}~\bibnamefont {Kerenidis}}, \ and\
  \bibinfo {author} {\bibfnamefont {J.}~\bibnamefont {Roland}},\ }\href@noop {}
  {\bibfield  {journal} {\bibinfo  {journal} {Quantum Inf. Comput.}\ }\textbf
  {\bibinfo {volume} {11}},\ \bibinfo {pages} {40} (\bibinfo {year}
  {2011})}\BibitemShut {NoStop}%
\bibitem [{\citenamefont {Lempel}(1975)}]{doi:10.1137/0204014}%
  \BibitemOpen
  \bibfield  {author} {\bibinfo {author} {\bibfnamefont {A.}~\bibnamefont
  {Lempel}},\ }\href {\doibase 10.1137/0204014} {\bibfield  {journal} {\bibinfo
   {journal} {SIAM J. Comput.}\ }\textbf {\bibinfo {volume} {4}},\ \bibinfo
  {pages} {175} (\bibinfo {year} {1975})}\BibitemShut {NoStop}%
\bibitem [{\citenamefont {Bernasconi}\ and\ \citenamefont
  {Codenotti}(1999)}]{10.1109/12.755000}%
  \BibitemOpen
  \bibfield  {author} {\bibinfo {author} {\bibfnamefont {A.}~\bibnamefont
  {Bernasconi}}\ and\ \bibinfo {author} {\bibfnamefont {B.}~\bibnamefont
  {Codenotti}},\ }\href {\doibase
  http://doi.ieeecomputersociety.org/10.1109/12.755000} {\bibfield  {journal}
  {\bibinfo  {journal} {IEEE Trans. Comput.}\ }\textbf {\bibinfo {volume}
  {48}},\ \bibinfo {pages} {345} (\bibinfo {year} {1999})}\BibitemShut
  {NoStop}%
\bibitem [{\citenamefont {Montanaro}\ and\ \citenamefont
  {Osborne}(2010)}]{montanaro2009ccxor}%
  \BibitemOpen
  \bibfield  {author} {\bibinfo {author} {\bibfnamefont {A.}~\bibnamefont
  {Montanaro}}\ and\ \bibinfo {author} {\bibfnamefont {T.}~\bibnamefont
  {Osborne}},\ }\href@noop {} {\enquote {\bibinfo {title} {On the communication
  complexity of {XOR} functions},}\ } (\bibinfo {year} {2010}),\ \Eprint
  {http://arxiv.org/abs/0909.3392v2} {arXiv:0909.3392v2} \BibitemShut {NoStop}%
\end{thebibliography}
\end{document}